\documentclass[a4paper,11pt]{article}
\usepackage{amsmath,amsthm}
\usepackage{amssymb,latexsym}
\usepackage{epsfig}
\usepackage{hyperref}
\usepackage{float}
\usepackage{mathpazo}
\usepackage{fullpage}
\usepackage{enumerate}
\usepackage{paralist}
\usepackage{times}

\newtheorem{theorem}{Theorem}
\newtheorem{corollary}[theorem]{Corollary}
\newtheorem{lemma}[theorem]{Lemma}

\newtheorem{definition}[theorem]{Definition}

\newtheorem{alg}[theorem]{Algorithm}

\newcommand{\myparagraph}[1]{\paragraph{#1.}}

\newcommand{\eps}{\varepsilon}
\newcommand{\epssdp}{\varepsilon_{\rm sdp}}
\newcommand{\bra}[1]{\langle #1|}
\newcommand{\ket}[1]{|#1\rangle}
\newcommand{\braket}[2]{\langle #1|#2\rangle}
\newcommand{\ketbra}[2]{\ket{#1}{\bra{#2}}}

\newcommand{\Tr}{Tr} 

\newcommand{\OPT}{{\rm OPT}}
\newcommand{\QMA}{{\rm QMA}}
\newcommand{\NP}{{\rm NP}}

\newcommand{\beq}{\begin{equation}}
\newcommand{\eeq}{\end{equation}}

\newcommand{\trace}{{\rm Tr}}

\newcommand{\norm}[1]{\left\|\,#1\,\right\|}       
\newcommand{\onorm}[1]{\norm{#1}_{\mathrm{1}}}      
\newcommand{\enorm}[1]{\norm{#1}_{\mathrm{2}}}      
\newcommand{\trnorm}[1]{\norm{#1}_{\mathrm {tr}}}  
\newcommand{\fnorm}[1]{\norm{#1}_{\mathrm {F}}}    
\newcommand{\snorm}[1]{\norm{#1}_{\mathrm {\infty}}}    

\newcommand{\set}[1]{{\left\{#1\right\}}}    
\newcommand{\ve}[1]{\mathbf{#1}}
\newcommand{\abs}[1]{\left\lvert #1 \right\rvert}
\newcommand{\optprod}{\OPT_P}

\newcommand{\optt}{\operatorname{OPT_2}}

\newcommand{\poly}{\operatorname{poly}}
\newcommand{\cc}{d^{\frac{k}{2}}}

\newcommand{\complex}{{\mathbb C}}
\newcommand{\reals}{{\mathbb R}}

\newcommand{\spa}[1]{\mathcal{#1}}

\newcommand{\klh}{MAX-$k$-local Hamiltonian}

\usepackage{color,graphicx}

\begin{document}

\title{Approximation algorithms for QMA-complete problems}
\author{Sevag Gharibian\footnote{David R. Cheriton School of Computer Science and Institute for Quantum Computing, University of Waterloo, Waterloo N2L 3G1, Canada. Supported by the Natural Sciences and Engineering Research Council of Canada (NSERC), NSERC Michael Smith Foreign Study Supplement, and by EU-Canada Transatlantic Exchange Partnership programme.}
\and
 Julia Kempe\footnote{Blavatnik School of Computer Science, Tel Aviv University, Tel Aviv 69978, Israel and CNRS \& LRI,
 University of Paris-Sud, Orsay, France. Supported
by an Individual Research Grant of the Israeli Science Foundation, by European Research Council (ERC) Starting Grant
QUCO and by the Wolfson Family Charitable Trust.}}
\date{}
\maketitle

\begin{abstract}
Approximation algorithms for classical constraint satisfaction problems are one of the main research areas in
theoretical computer science. Here we define a natural approximation version of the QMA-complete local Hamiltonian
problem and initiate its study. We present two main results. The first shows that a non-trivial approximation ratio can
be obtained in the class $\NP$ using product states. The second result (which builds on the first one), gives a
polynomial time (classical) algorithm providing a similar approximation ratio  for dense instances of the problem. The
latter result is based on an adaptation of the ``exhaustive sampling method" by Arora et al.~\cite{AKK99} to the
quantum setting, and might be of independent interest.
\end{abstract}


\section{Introduction and Results}\label{scn:intro}
In the last few years, the quantum analog of the class NP, the class QMA~\cite{KSV02}, has been extensively studied,
and several QMA-complete problems have been found~\cite{L06,B06,LCV07,BS07,R09,JGL09,SV09,WMN10}. Arguably the
most important QMA-complete problem is the $k$-local Hamiltonian problem~\cite{KSV02,KR03,OT05,KKR06,AGIK09}. Here, the
input is a set of Hamiltonians (Hermitian matrices), each acting on at most $k$-qubits each. The task is to determine
the largest eigenvalue of the sum of these Hamiltonians. This problem generalizes  the central NP-hard problem
MAX-$k$-CSP, where we are given a set of Boolean constraints on $k$ variables each, with the goal to satisfy as many
constraints as possible. The local Hamiltonian problem is of significant interest to complexity theorists and to
physicists studying properties of physical systems alike (e.g.~\cite{BV05,ADKLLR07,BDOT08,AALV09,CV09,LLMSS10,SC10}).

Moving to the classical scenario, the theory of NP-completeness is one of the great success stories of classical
computational complexity~\cite{AB90}. It was soon realized that many natural optimization problems are NP-hard, and are hence unlikely to have polynomial time algorithms. A natural question (both in theory and in practice) is to look for polynomial time algorithms that produce
solutions that are close to optimum. More precisely, one says that an algorithm achieves an \emph{approximation ratio}
of $c \in [0,1]$ for a certain maximization problem if on all inputs, the value of the algorithm's output is at least
$c$ times that of the optimum solution. The closer $c$ is to $1$, the better the approximation. The
investigation of approximation algorithms is, after decades of heavy research, still a very active area (e.g.,~\cite{H97,V01}). For many central NP-hard problems, tight polynomial time approximation algorithms
are known.

In the context of QMA-complete problems, it is thus natural to search for approximation algorithms for these
problems, and in particular for the local Hamiltonian problem. The question we address here is: \emph{How well can one efficiently approximate the $k$-local Hamiltonian problem?}

It should be noted that a large host of heuristics has been developed in the physics community to
approximate properties of local Hamiltonian systems (see, e.g.,~\cite{CV09} for a survey) and this area
is extremely important in the study of physical systems. However, the systematic complexity theoretic study of
approximation algorithms for QMA-complete problems is still very much in its infancy, and our work is one of the first
steps in this research direction. We note that there has been a lot of interest in recent years~\cite{AALV09,A06} in
establishing a so-called quantum PCP theorem (e.g.~\cite{AS98,ALMSS98}), which amounts to showing that for some
constant $c<1$ close enough to $1$, approximating the $k$-local Hamiltonian (or related problems) to within $c$ is
QMA-hard. Our results can also be seen as a natural continuation of that investigation.

\paragraph{Our results:}

Let us start by precisely defining the optimization version of the local
Hamiltonian problem, which is parameterized by two integers $k$ and $d$, which we always think of as constants.

\begin{definition}[MAX-$k$-local Hamiltonian problem on $d$-level systems (qudits)]  An instance of the problem
consists of a collection of $\binom{n}{k}$ Hermitian matrices, one for each subset of $k$ qudits. The matrix
$H_{i_1,\ldots,i_k}$ corresponding to some $1 \leq i_1 \le \cdots \le i_k \le n$ is assumed to act on those
qudits\footnote{Terms acting on less than $k$ qudits can be
incorporated by tensoring them with the identity.}, to be positive semidefinite, and to have operator norm at most $1$.
We call any pure or mixed state $\rho$ on $n$ qudits an \emph{assignment} and define its {\em value} to be $\Tr H \rho$
where $H=\sum_{i_1, \ldots ,i_k} H_{i_1,\ldots,i_k}$. The goal is to find the largest eigenvalue of $H$ (denoted
$\OPT$), or equivalently, the maximum value obtained by an assignment. We say that an algorithm provides an
\emph{approximation ratio} of $c \in [0,1]$ if for all instances, it outputs a value that is between $c \cdot \OPT$ and
$\OPT$.
\end{definition}


This definition, we believe, is the natural quantum analog of the MAX-$k$-CSP problem. We note that it differs slightly
from the usual definition of the $k$-local Hamiltonian problem. Namely, we consider maximization (as opposed to
minimization), and also restrict the terms of $H$ to be positive semidefinite, and have norm at most $1$. As long as
one considers the \emph{exact} problem, these assumptions are without loss of generality, and do not affect the
definition, as seen by simply scaling the Hamiltonians and adding multiples of identity as necessary. However,
when dealing with the \emph{approximation} version, these assumptions are important for the problem to make sense; for
instance, one cannot meaningfully talk about approximation ratios if the optimum can take both negative and positive
values. That is why we require the terms to be positive semidefinite. The requirement that the terms have operator norm
at most $1$ does not affect the problem and later allows us to conveniently define dense instances.
Finally, changing the maximization to a minimization would lead to an entirely different approximation problem: the
quantum analogue of MIN-CSP (e.g.~\cite{KSTW01}). Minimization problems
are, generally speaking, harder than maximization problems, and we leave this research direction for future work. 

Before stating our results, we remark that there is a trivial way to get a $d^{-k}$-approximation for MAX-$k$-local
Hamiltonian. Observe that the maximally mixed state has at least $d^{-k}$ overlap with the reduced
density matrix of the optimal assignment on any $k$ particles. A similar thing holds classically, where a random
assignment gives (in expectation) a $d^{-k}$ approximation of MAX-$k$-CSP. We now describe our two main results.

\myparagraph{Approximation by product states} One inherently quantum property of the local Hamiltonian problem is the
fact that the optimal state might in general be highly entangled (and hence not efficiently describable in polynomial
time). This is why we do not require  outputting the assignment itself in the above definition. If, however, the
optimal assignment (or some other good assignment) was guaranteed to be a {\em tensor product state}, then we could describe it efficiently. The following theorem shows just that.

\begin{theorem}\label{thm:approxprodstate}
For an instance of MAX-$k$-local Hamiltonian with optimal value $\OPT$, there is a (pure) product state assignment that has
value at least $\OPT/d^{k-1}$.
\end{theorem}
This result is {\em tight} for product states in the case of $2$-local Hamiltonians (we remark that $2$-local Hamiltonians are often the most relevant case from a physics perspective). For example, consider the Hamiltonian on $2$-qubits
that projects onto the EPR state $\frac{1}{\sqrt{2}}(\ket{00}+\ket{11})$. It is easy to see that no product state
achieves value more than $1/2$. For general $k$, we can only show that product states cannot achieve an approximation
ratio greater than $1/d^{\lfloor k/2 \rfloor}$ (see Sec.~\ref{scn:sepopt}).


If we could efficiently find the best product state assignment, we would obtain an algorithm achieving a non-trivial
$d^{-k+1}$ approximation ratio. Unfortunately, this problem is NP-complete, since it would allow one to solve (e.g.)
the special case of MAX-$k$-SAT, and hence we do not have such an algorithm. Still, the theorem has the following
interesting implication: It shows that unless $\NP=\QMA$, approximating the local Hamiltonian problem to within a factor less than $d^{-k+1}$ is not QMA-hard. This follows simply because product states have polynomial size
classical descriptions.

\myparagraph{A polynomial time approximation algorithm for dense instances}

Our second result gives a classical polynomial time approximation algorithm for {\em dense} instances of the local
Hamiltonian problem. This result is perhaps our technically most challenging one, and we hope the techniques we develop
might turn out useful elsewhere.

Dense instances of classical constraint satisfaction problems have been studied in depth~(see
e.g.~\cite{V96,FK96,GGR98,AKK99,VK00,AVKK02,BVK03,VKKV05}). Our result is inspired by work of Arora et al.~\cite{AKK99}
who provide a polynomial time approximation scheme, or PTAS (i.e., an efficient $1-\eps$ approximation algorithm for
any fixed $\eps>0$), for several types of dense constraint satisfaction problems. In the classical case, dense (for
$2$-local constraints) simply means that the average degree in the constraint graph is $\Omega(n)$, or equivalently,
that the optimum is $\Omega(n^2)$.  In
analogy, we define an instance of MAX-$k$-local Hamiltonian to be {\em dense} if $\OPT=\Omega(n^k)$, or equivalently,
if $\Tr(H \frac{Id}{d^n})=\Omega(n^k)$.\footnote{The equivalence  follows from the fact that the mixed state assignment $Id/d^n$ has value between $\OPT$ and $\OPT/d^k$.}

It is not hard to see that the (exact) dense local Hamiltonian problem remains QMA-hard (see Sec.~\ref{sscn:finalptas}).
The
dense case might be of practical interest to physicists who study systems of particles by incorporating all possible
interactions between them. Note that such instances are dense even if the interactions between particles are weak, so long as the interaction strengths are constants independent of $n$. Our second main result is the following:

\begin{theorem}\label{thm:approxdense}
For all $\eps >0$ there is a polynomial time $(1/d^{k-1} -\eps)$ approximation algorithm for the dense MAX-$k$-local
Hamiltonian problem over qudits.
\end{theorem}
Thm.~\ref{thm:approxdense} follows immediately by combining Thm.~\ref{thm:approxprodstate} with the following
theorem, which gives an approximation scheme for the problem of optimizing over the set of product states.

\begin{theorem}\label{thm:ptasdense}
For all $\eps>0$ there is a polynomial time algorithm for dense MAX-$k$-local Hamiltonian that outputs a
product state assignment with value within $1-\eps$ of the value of the best product state assignment.
\end{theorem}

\paragraph{Proof ideas and new tools:}
The proofs of Thm.~\ref{thm:approxprodstate} and Thm.~\ref{thm:ptasdense} are independent and employ different
techniques. To show the product state approximation guarantee, we show a slightly stronger statement: For {\em any}
assignment $\ket{\Psi}$, there is a way to construct a product assignment of at least $d^{-k+1}$ its value. The proof is constructive (given $\ket{\Psi}$): we use a type of recursive Schmidt decomposition of $\ket{\Psi}$ to obtain a mixture
of product states whose value is off by at most the desired approximation factor (see Sec.~\ref{scn:prodratio}).

Our second result is technically more challenging and introduces a few new ideas to this problem, inspired by work of
Arora et al.~\cite{AKK99} in the classical setting. We illustrate the main ideas for MAX-$2$-local Hamiltonian on $n$
qubits. Recall that our goal is to find a PTAS for the local Hamiltonian problem {\em over product states}. The value
of the optimal product state assignment, $\optprod$, can be written
\begin{equation}
      \optprod\hspace{2mm} = \hspace{2mm}\max \quad\sum_{i=1}^n \sum_{j \in N(i)} \Tr (H_{i,j} (\rho_i \otimes \rho_j))
    \quad\mbox{s.t.}\quad \rho_i\succeq 0 \mbox{ and }\trace(\rho_i)=1\quad\mbox{for } 1\leq i\leq n,\label{eqn:introprogram}
\end{equation}
where $N(i)$ is the set of indices $j$ for which a local Hamiltonian term $H_{i,j}$ is present. We might call this a
{\em quadratic} semidefinite program, as the maximization is quadratic in the $\rho_i$ (and as such not efficiently
solvable). Note, however, that if the terms in the maximization were {\em linear}, then we would obtain a semidefinite
program (SDP), which is efficiently solvable~\cite{GLS93}. To ``linearize'' our optimization, we use the ``exhaustive
sampling method" developed by Arora et al.~\cite{AKK99} (a method which was later key in many developments in property
testing, e.g.~\cite{GGR98}). We write each Hamiltonian term in a basis that separates its two qubits, for instance the
Pauli basis $\{\sigma_0, \sigma_1, \sigma_2, \sigma_3\}$, $H_{i,j}=\sum_{k,l=0}^3 \alpha^{ij}_{kl} \sigma_k \otimes
\sigma_l$. For $i=1,\ldots,n$ and $k=0,1,2,3$, define
\[
    c_k^i := \sum_{j \in N(i)} \sum_l \alpha_{kl}^{ij} \Tr (\sigma_l \rho_j).
\]
If we knew the values of $c_k^i$ for the optimal $\rho_i$, then solving the SDP below would yield the optimal $\rho_i$:
\begin{align}
    \max \quad \sum_{i=1}^n \sum_{k=0}^3 c_k^i \Tr (\sigma_k \rho_i)\quad
    \mbox{s.t.}\quad& \rho_i\succeq 0\mbox{ and } \trace(\rho_i)=1\hspace{11mm} \mbox{for } 1\leq i\leq n,\label{eqn:introdecomp}\\
    &\sum_{j \in N(i)} \sum_l \alpha_{kl}^{ij} \Tr (\sigma_l \rho_j)=c_k^i\hspace{5mm} \mbox{for } 1\leq i\leq n\mbox{ and } 0\leq k\leq 3.\nonumber
\end{align}
Of course, this reasoning is circular, as in order to obtain the $c_k^i$ we need the optimal $\rho_i$. The crucial idea
is now to use {\em sampling} to {\em estimate} the $c_k^i$. More precisely, assume for a second that we could sample
$O(\log n)$ of the $\rho_i$ randomly from the optimal assignment. Then, by standard sampling bounds, with high
probability over the choice of the sampled qubits we can estimate all the $c_k^i$ to within an additive error $\pm \eps
n$ for some $\eps$. If we had these estimates $a_k^i$ for the $c_k^i$, we could solve the SDP above with the slight
modification that the last constraint should be $a_k^i -\eps n \leq \sum_{j \in N(i)} \sum_l \alpha_{kl}^{ij} \Tr
(\sigma_l \rho_j)\leq a_k^i +\eps n$. With high probability over the sampled qubits, this SDP will give a solution that
is within an additive $\eps n^2$ of the optimal one (more subtle technicalities and all calculations can be found in
Sec. \ref{scn:sepopt}). Moreover, it is possible to derandomize the sampling procedure to obtain a deterministic
algorithm (Sec.~\ref{sscn:finalptas}).

Of course, we are still in the realm of wishful thinking, because in order to sample from the optimal solution, we
would need to know it, which is precisely what we set out to do. However, the number of qubits we  wish to sample is
only \emph{logarithmic} in the input size. Thus, to simulate the sampling procedure, we can pick a random subset of
$O(\log n)$ qubits, and simply \emph{iterate} through all
possible assignments on them (with an appropriate $\delta$-net over the density matrices, which incurs a small
additional error) in polynomial time! Our algorithm then runs the SDP for each iteration, and we are guaranteed that at
least one iteration will return a solution within $\eps n^2$ of the optimal one. Because the denseness assumption
guarantees that $\optprod$ is $\Omega(n^2)$, our additive approximation turns into a factor $(1-\eps)$-approximation,
as desired. All details, the runtime of the algorithm and error bounds for the general $k$-local case on qudits
are given in Sec.~\ref{scn:sepopt}.

\paragraph{Previous and related work:}
We note that many heuristics have been developed in the physics community to
approximate properties of local Hamiltonian systems and this area is extremely important in the study of physical
systems~(e.g.~\cite{W92,W93,OR95,RO97,S05,PWKH98,CV09}). Our focus here is, however, on rigorous bounds on the
approximation guarantee of algorithms for the general problem. In this area, to our knowledge, few results are known. In
a first result on polynomial time approximation algorithms, Bansal, Bravyi and Terhal \cite{BBT09} give a PTAS for a
special class of the local Hamiltonian problem, so called quantum Ising spin glasses, for the case where the instance
is on a planar graph (and in particular of bounded degree). This PTAS is obtained by dividing the graph into constant
size chunks, which can be solved directly, and ignoring the constraints between chunks (this incurs an error
proportional to the number of such constraints, which is small because the graph is planar).
More recently, there has
been work proving rigorous approximations to ground states of one-dimensional quantum systems under well-defined
conditions using techniques such as density matrix renormalization group~\cite{AAI10,SC10}. To our
knowledge, we are the first to establish a bound on the approximation factor by optimizing over the set of product
states.

\paragraph{Discussion and open questions:}

Our two results give approximations to the local Hamiltonian problem. Although at first glance, our approximation ratio of $1/d^{k-1}$ may appear an incremental improvement over the trivial random assignment strategy, there are two important notes that should be kept in mind: The first is that many classical NP-hard problems, such as MAX-3-SAT (a special case of MAX-$k$-CSP where each constraint is the disjunction (``OR") of $k$ variables or their negation), are \emph{approximation resistant}~(e.g.~\cite{H07,AM08}), meaning that unless P$=$NP, there do not even exist non-trivial approximation ratios beyond the random assignment strategy. For example, for MAX-3-SAT it is NP-hard to do better than the approximation ratio of $7/8$ achieved by random assignment~\cite{Ha97}. Thus, showing the existence of a non-trivial approximation ratio is typically a big step in the classical setting. Moreover, it could have been conceivable that for \klh, analogously to MAX-3-SAT, outperforming the random assignment strategy would have been \emph{QMA-hard}. Yet our results show that unless NP$=$QMA, this is not the case. The second important note that should be kept in mind is that the currently \emph{best} approximation algorithm for MAX-$k$-CSP gives an approximation ratio of only about $0.44k/2^k$ for $k>2$~\cite{CMM07} (for $k=2$, one can achieve $0.874$~\cite{LLZ02}) and this is, moreover, essentially the best possible under a plausible complexity
theoretic conjecture (namely, the unique games conjecture~\cite{K02})~\cite{T98,H05,ST06,AM08}. This is to be contrasted with our $2/2^k$-approximation ratio for the case of $d=2$ (i.e. qubit systems), which we show can be achieved by product state assignments for \emph{arbitrary} (i.e. even non-dense) \klh~instances. This raises the important open question: is our approximation ratio tight?

Our product state approximation shows that approximating the local Hamiltonian problem to within $d^{-k+1}$
is in NP. It would be interesting to know if this approximation ratio could also be achieved in polynomial time. If
not, it might lead to an intriguing state of affairs where for low approximation ratios the problem is efficiently
solvable, for medium ratios it is in NP but not efficiently solvable, and for high ratios it is QMA-hard (assuming
a quantum PCP theorem exists).

To obtain our results for the case of dense local Hamiltonians, we have introduced the exhaustive sampling technique of
Arora et al.~\cite{AKK99} to the setting of low-degree semidefinite programs. We linearize such programs using
exhaustive sampling in combination with a careful analysis of the error coming from working with $\delta$-nets on
density matrices. We remark that it seems we cannot simply apply the results of~\cite{AKK99} for \emph{smooth
Polynomial Integer Programs} as a black-box to our setting. This is due to our
aforementioned need for a $\delta$-net, as well as the requirement that our assignment be a positive semidefinite
operator. We address the latter issue by extending the techniques of~\cite{AKK99} to the realm of positive semidefinite
programs by introducing the notion of ``degree-$k$ inner products'' over Hermitian operators to generalize the concept of degree-$k$ polynomials over real numbers, and performing the more complex analysis that ensues.
We hope that this technique will be of much wider applicability, particularly considering the growing use of semidefinite programs in numerous areas of quantum computing and information (e.g.~\cite{doherty04a,JJUW10,LMRS10}).

Another open question is whether similar ideas can be used to approximate other QMA-complete problems,
such as the consistency problem~\cite{L06}. Moreover, can we obtain polynomial time algorithms without the denseness
assumption? And are there special cases of the local Hamiltonian problem for which there is a PTAS (other than for planar
Ising spin glasses \cite{BBT09})? Of course, we do not expect a PTAS for all instances of the local
Hamiltonian problem, as this would contradict known hardness results for special classical cases of the problem.
However, perhaps there exist other classes of physically relevant instances of the problem for which a PTAS does exist.

\paragraph{Structure of this paper:} In Sec.~\ref{scn:prodratio}, we prove our result on product state approximations (Thm.~\ref{thm:approxguarantee} and the ensuing proof of Thm.~\ref{thm:approxprodstate}), show its tightness in the $2$-local case and provide the upper bound of $d^{-\lfloor k/2\rfloor}$ for the best possible
approximation by product states. Sec.~\ref{scn:sepopt} gives our polynomial time approximation algorithm and develops the
general sampling and SDP-based technique we use. It also shows that the dense local Hamiltonian problem remains
QMA-complete.

\paragraph{Notation:} We use $A\succeq 0$ to say operator $A$ is positive semidefinite, and denote by $L(\spa{X})$, $H(\spa{X})$, and $D(\spa{X})$ the sets of linear, Hermitian, and density operators acting on complex Euclidean space $\spa{X}$, respectively. We denote the Frobenius and operator norms of $A\in L(\spa{X})$ as $\fnorm{A}=\sqrt{\trace(A^\dagger A)}$ and $\snorm{A}=\max_{\ket{x}\in\spa{X}\mbox{ s.t. } \enorm{x}= 1}\enorm{A\ket{x}}$, respectively.

\section{Product states yield a $1/d^{k-1}$-approximation for qudits}\label{scn:prodratio}

We now show that product state assignments achieve a non-trivial approximation ratio for \klh~, i.e. Thm~\ref{thm:approxprodstate}. The heart of our approach is what we call the \emph{Mixing Lemma} (Lem.~\ref{l:mixinglemma}), which we use to prove Thm.~\ref{thm:approxguarantee}. Thm.~\ref{thm:approxprodstate} will then easily follow. At the end of the section, we discuss the tightness of the approximation guarantee given by Thm.~\ref{thm:approxprodstate}. We begin with two definitions.

\begin{definition}[Recursive Schmidt Decomposition (RSD)] \textup{We define the \emph{recursive Schmidt decomposition} of a state $\ket{\psi}\in(\complex^d)^{\otimes n}$ as the expression obtained by recursively applying the Schmidt decomposition on each qudit from $1$ to $n-1$ inclusive\footnote{This definition is relative to some fixed ordering of the qudits. The specific choice of ordering is unimportant in our scenario, as any decomposition output by such a process suffices to prove Thm.~\ref{thm:approxprodstate}.}. For example, the RSD for $3$-qubit $\ket{\psi}$ is
\begin{equation*}
    \ket{\psi} = \alpha_1 \ket{a_1}\otimes \left(\beta_1\ket{b_1}\ket{c_1}+\beta_2\ket{b_2}\ket{c_2}\right)+\alpha_2 \ket{a_2}\otimes (\beta^\prime_1\ket{{b^\prime}_1}\ket{{c^\prime}_1}+\beta^\prime_2\ket{{b^\prime}_2}\ket{{c^\prime}_2}),
\end{equation*}
for $\alpha_1^2+\alpha_2^2=\beta_1^2+\beta_2^2={\beta^\prime_1}^2+{\beta^\prime}_2^2=1$, $\set{\ket{a_i}}_i$ an orthonormal basis for qubit $1$, $\set{\ket{b_i}}_i$ and $\set{\ket{{b^\prime}_i}}_i$ orthonormal bases for qubit $2$, and $\set{\ket{c_i}}_i$ and $\set{\ket{{c^\prime}_i}}_i$ orthonormal bases for qubit $3$.}
\end{definition}

\begin{definition}[Schmidt cut]
    \textup{For any $\ket{\psi}\in(\complex^d)^{\otimes n}$ with Schmidt decomposition $\ket{\psi}=\sum_{i=1}^{d} \alpha_i \ket{w_i}\ket{v_i}$, where $\ket{w_i}\in\complex^d$ and $\ket{v_i}\in(\complex^d)^{\otimes n-1}$, and for any $\ket{\phi}\in(\complex^d)^{\otimes m}$, we refer to the expansion $\ket{\phi}\otimes\left(\sum_{i=1}^{d} \alpha_i \ket{w_i}\ket{v_i}\right)$ as the \emph{Schmidt cut} at qudit $m+1$. We say that a projector $\Pi$ \emph{crosses} this Schmidt cut if $\Pi$ acts on qudit $m+1$ and at least one qudit $i\in\set{m+2,\ldots,m+n}$.}
\end{definition}

The heart of our approach is the following Mixing Lemma, which provides, for \emph{any} assignment $\ket{\psi}\in(\complex^d)^{\otimes n}$, an explicit construction through which the entanglement across the first Schmidt cut of $\ket{\psi}$ can be eliminated, while maintaining at least a $(1/d)$-approximation ratio relative to the value $\ket{\psi}$ achieves against any local Hamiltonian $H\in H((\complex^d)^{\otimes n})$.

\begin{lemma}[Mixing Lemma]\label{l:mixinglemma}
Given state $\ket{\psi}$ on $n$ qudits with Schmidt cut on qudit $1$ given by $\ket{\psi} = \sum_{i=1}^{d} \alpha_i \ket{w_i}\ket{v_i}$, where $\ket{w_i}\in\complex^d$ and $\ket{v_i}\in(\complex^d)^{\otimes n-1}$, define $\rho := \sum_{i=1}^{d}\alpha_i^2\ketbra{w_i}{w_i}\otimes\ketbra{v_i}{v_i}$. Then, given projector $\Pi$ acting on some subset $\mathcal{S}$ of the qudits, if $\Pi$ crosses the Schmidt cut, then $\trace(\Pi\rho)\geq\frac{1}{d}\trace(\Pi\ketbra{\psi}{\psi})$. Otherwise, $\trace(\Pi\rho)=\trace(\Pi\ketbra{\psi}{\psi})$.
\end{lemma}
\begin{proof}
    Case $2$ follows easily by noting that the given Schmidt decomposition of $\ket{\psi}$ implies $\trace_1(\rho)=\trace_1(\ketbra{\psi}{\psi})$ and $\trace_{2,\ldots,n}(\rho)=\trace_{2,\ldots,n}(\ketbra{\psi}{\psi})$. To prove case $1$, we observe by straightforward expansion that
    \begin{equation}
        \trace(\Pi\ketbra{\psi}{\psi}) = \trace(\Pi\rho) + \sum_{i< j}\alpha_i\alpha_j\bra{w_i}\bra{v_i}\Pi\ket{w_j}\ket{v_j} + \alpha_i\alpha_j\bra{w_j}\bra{v_j}\Pi\ket{w_i}\ket{v_i}.\label{eqn:mixlem1}
    \end{equation}
    Then, by defining for each $i$ vector $\ket{a_{i}} := \alpha_i\Pi\ket{w_i}\ket{v_i}$, we have that
    \begin{equation*}
        \sum_{i< j}\alpha_i\alpha_j\bra{w_i}\bra{v_i}\Pi\ket{w_j}\ket{v_j} + \alpha_i\alpha_j\bra{w_j}\bra{v_j}\Pi\ket{w_i}\ket{v_i}= \sum_{i<j}\braket{a_{i}}{a_{j}}+\braket{a_{j}}{a_{i}}.
    \end{equation*}
    Applying the facts that $\Pi^2=\Pi$ and $\braket{a}{b} + \braket{b}{a} \leq \enorm{\ket{a}}^2 + \enorm{\ket{b}}^2$ for $\ket{a},\ket{b}\in(\complex^d)^{\otimes n}$ implies
    \begin{equation*}
        \sum_{i<j}\braket{a_{i}}{a_{j}}+\braket{a_{j}}{a_{i}}\leq\sum_{i< j}\enorm{\ket{a_{i}}}^2 + \enorm{\ket{a_{j}}}^2
        =(d-1)\sum_{i}\alpha_i^2\bra{w_i}\bra{v_i}\Pi\ket{w_i}\ket{v_i}
        =(d-1)\trace(\Pi\rho),
    \end{equation*}
    from which the claim follows.
\end{proof}

The following simple extension of Lem.~\ref{l:mixinglemma} simplifies our proof of Thm.~\ref{thm:approxguarantee}.

\begin{corollary}\label{cor:mixinglemma}
    Define $\ket{\psi^\prime}:=\ket{\phi}\otimes \ket{\psi}$, where $\ket{\phi}\in(\complex^d)^{\otimes m}$ for $m>0$ and $\ket{\psi}$ is defined as in Lem.~\ref{l:mixinglemma}, and let $\rho\in D(\complex^d)^{\otimes n}$ be obtained from $\ket{\psi}$ as in Lem.~\ref{l:mixinglemma}. Then, for any projector $\Pi$ acting on a subset $\mathcal{S}$ of the qudits, if $\Pi$ crosses the Schmidt cut of $\ket{\psi^\prime}$ at qudit $m+1$, we have  $\trace(\Pi\ketbra{\phi}{\phi}\otimes\rho)\geq\frac{1}{d}\trace(\Pi\ketbra{\psi^\prime}{\psi^\prime})$. Otherwise, $\trace(\Pi\ketbra{\phi}{\phi}\otimes\rho)=\trace(\Pi\ketbra{\psi^\prime}{\psi^\prime})$.
\end{corollary}

\begin{proof}
     Immediate by applying the proof of Lem.~\ref{l:mixinglemma} with the following modifications: (1) Define $\ket{a_{i}} := \alpha_i\Pi\ket{\phi}\ket{w_i}\ket{v_i}$, and (2) if $\mathcal{S}\subseteq\set{1,\ldots,m}\cup\set{m+2,\ldots,m+n}$ (i.e. this is one of two ways for $\Pi$ not to cross the cut --- the other way is for $\mathcal{S}\subseteq \set{1,\ldots,m+1}$), observe that by the same arguments as in Lem.~\ref{l:mixinglemma} for case $2$ and the product structure between $\ket{\phi}$ and $\ket{\psi}$ in $\ket{\psi^\prime}$ that $\trace_{m+1}(\ketbra{\phi}{\phi}\otimes\rho)=\trace_{m+1}(\ketbra{\psi^\prime}{\psi^\prime})$.
\end{proof}

Lemma~\ref{l:mixinglemma} shows that the state $\rho$ obtained by \emph{mixing} the $d$ Schmidt vectors of $\ket{\psi}$, as opposed to taking their \emph{superposition}, suffices to achieve a $(1/d)$-approximation across the first Schmidt cut. By iterating this argument over \emph{all} $n-1$ Schmidt cuts, we now prove that a mixture of all (product) states appearing in the RSD of $\ket{\psi}$ achieves an approximation ratio of $1/d^{k-1}$.
\begin{theorem}\label{thm:approxguarantee}
    For any $n$-qudit assignment $\ket{\psi}$ with RSD $\ket{\psi}=\sum_{i=1}^{d^{n-1}}\sqrt{p_i}\ket{\phi_i}$, where $\sum_i p_i = 1$ and $\set{\ket{\phi_i}}_{i=1}^{d^{n-1}}$ is a set of orthonormal product vectors in $(\complex^d)^{\otimes n}$, define $\rho := \sum_{i=1}^{d^{n-1}} p_i \ketbra{\phi_i}{\phi_i}$. Then, for any projector $\Pi$ acting on some subset $\mathcal{S}\subseteq\set{1,\ldots,n}$ of qudits with $\abs{\mathcal{S}}= k$, we have
$        \trace(\Pi\rho)\geq\frac{1}{d^{k-1}}\trace(\Pi\ketbra{\psi}{\psi})$.
\end{theorem}
\begin{proof}
 Let $\Pi$ be a projector with $\abs{\mathcal{S}}= k$, and define $\ve{c}\in\set{0,1}^{n-1}$ such that $\ve{c}(j)=1$ iff $\Pi$ crosses the Schmidt cut at qudit $j$. For example, if $\Pi$ acts on qudits $\set{1,2}$, then $\ve{c}=(1,0,\ldots,0)$. Note that in general $\onorm{\ve{c}}=k-1$. We proceed by iteratively stepping through each Schmidt cut in the RSD of $\ket{\psi}$. Let $\rho^{(0)}:=\ketbra{\psi}{\psi}$, and consider first the cut at qudit $1$, i.e. $\ket{\psi} = \sum_{i=1}^{d} \alpha_i \ket{w_i}\ket{v_i}$, for $\ket{w_i}\in\complex^d$ and $\ket{v_i}\in(\complex^d)^{\otimes n-1}$. Defining $\rho^{(1)} := \sum_{i=1}^{d}\alpha_i^2\ketbra{w_i}{w_i}\otimes\ketbra{v_i}{v_i}$, we have by Lem.~\ref{l:mixinglemma} that
\begin{equation}
   \trace(\Pi\ketbra{\psi}{\psi})\leq d^{\ve{c}(1)} \trace(\Pi\rho^{(1)}),\label{eqn:recurse1}
\end{equation}
i.e. we lose a factor of $1/d$ iff $\Pi$ crosses the first cut.

Moving on to the second Schmidt cut, consider the state $\ket{w_1}\ket{v_1}\in\complex^d\otimes(\complex^d)^{\otimes n-1}$ appearing in the expression for $\rho^{(1)}$. Observe that it satisfies the preconditions for Cor.~\ref{cor:mixinglemma} with $m=1$. Hence, via Cor.~\ref{cor:mixinglemma} there exists a state $\sigma_1$ acting on qudits $\set{2,\ldots,n}$ such that $\trace(\Pi\ketbra{w_1}{w_1}\otimes\ketbra{v_1}{v_1})\leq d^{\ve{c}(2)}\trace(\Pi\ketbra{w_1}{w_1}\otimes\sigma_1)$. We can analogously find states $\sigma_i$ corresponding to $\ket{w_i}\ket{v_i}$ for all $1\leq i\leq d$. Thus,
\begin{equation}
    \trace(\Pi\rho^{(1)}) = \sum_{i=1}^d\alpha_i^2\trace(\Pi\ketbra{w_i}{w_i}\otimes\ketbra{v_i}{v_i})
                    \leq d^{\ve{c}(2)}\left[\sum_{i=1}^d\alpha_i^2\trace(\Pi\ketbra{w_i}{w_i}\otimes\sigma_i) \right].\label{eqn:iter2}
\end{equation}
Hence, by defining $\rho^{(2)} := \sum_{i=1}^d\alpha_i^2\ketbra{w_i}{w_i}\otimes\sigma_i$, we have via Eqns.~(\ref{eqn:recurse1}) and~(\ref{eqn:iter2}) that
\begin{equation*}
    \trace(\Pi\ketbra{\psi}{\psi})\leq d^{\ve{c}(1)+\ve{c}(2)}\trace(\Pi\rho^{(2)}).
\end{equation*}
Since by Cor.~\ref{cor:mixinglemma}, the $\sigma_i$ are mixtures of Schmidt vectors from the second Schmidt cut, we can now iteratively apply the same procedure to the (at most $d^2$) pure states appearing in the expression for $\rho^{(2)}$ when considering the third Schmidt cut. Note in particular that each of these terms will have a product structure between qudits $\set{1,2}$ and $\set{3,\ldots,n}$, as required by Cor.~\ref{cor:mixinglemma} for the next iteration.

More generally, when considering the $p$th Schmidt cut, we apply Cor.~\ref{cor:mixinglemma} with $m=p-1$ to each of the at most $d^{p-1}$ terms appearing in the expansion of $\rho^{(p-1)}$. We continue iterating in this fashion until we have exhausted all $n-1$ Schmidt cuts, at which point the resulting mixture $\rho^{(n-1)}$ we are left with is in fact the $\rho$ from the statement of the claim (seen by noting that our procedure essentially iteratively computes the RSD of $\ket{\psi}$, mixing the Schmidt vectors it computes at each step). Moreover, due to the repeated application of Cor.~\ref{cor:mixinglemma}, we have
\begin{equation}
    \trace(\Pi\ketbra{\psi}{\psi})\leq d^{\onorm{\ve{c}}}\trace(\Pi\rho^{(n-1)}).
\end{equation}
Recalling that $\onorm{\ve{c}}=k-1$ completes the proof.
\end{proof}
\paragraph{Proof of Thm.~\ref{thm:approxprodstate}:} Simply apply Thm.~\ref{thm:approxguarantee} to each projector in the spectral decompositions of each (positive semidefinite) $H_i$ in our \klh~instance $H=\sum_i H_i$, and let $\ket{\psi}$ denote the optimal assignment for $H$. It is important to note that we can exploit Thm.~\ref{thm:approxguarantee} in this fashion due to the fact that the $\rho$ constructed by Thm.~\ref{thm:approxguarantee} is \emph{independent} of the projector $\Pi$ --- i.e. for any fixed $\ket{\psi}$ and $k$, the state $\rho$ provides the same approximation ratio against \emph{any} $k$-local projector $\Pi$ encountered in the spectral decompositions of the $H_i$. Finally, note that one can find a \emph{pure} product state achieving this approximation guarantee since $\rho$ is a convex mixture of pure product states.

\paragraph{Upper bound of $d^{-\lfloor \frac{k}{2}\rfloor}$ for product state approximations.}
Is the result of Thm.~\ref{thm:approxprodstate} tight? In the case of MAX-$2$-local Hamiltonian on qudits, yes --- consider a single clause projecting onto the maximally entangled state $\frac{1}{\sqrt{d}}\sum_i \ket{ii}$, for which a product state achieves value at most $1/d$. On the other hand, for MAX-$3$-local Hamiltonian on qubits, the worst case clause for a $3$-qubit product state assignment is the projector onto the state $\ket{W}=\frac{1}{\sqrt{3}}(\ket{001}+\ket{010}+\ket{100})$~\cite{TWP09}. But here product states achieve value $4/9$~\cite{WG03}, implying the bound of $1/4$ from Thm.~\ref{thm:approxprodstate} is not tight.

A simple construction shows that the true optimal ratio is upper bounded by $d^{-\lfloor\frac{k}{2}\rfloor}$. To see this, consider a single clause which
is the tensor product of maximally entangled bipartite states\footnote{For odd $k$, we assume the odd qudit out projects onto the identity.}. For example, for $n=4$, consider the clause
$\ketbra{\phi^+}{\phi^+}\otimes\ketbra{\phi^+}{\phi^+}$, where
$\ket{\phi^+}=\frac{1}{\sqrt{2}}(\ket{00}+\ket{11})$. The maximum value
a product state can attain is $1/4$, as claimed. In the qubit setting ($d=2$), one can further improve this construction for odd $k$ by replacing the term $\ketbra{\phi^+}{\phi^+}\otimes I$ on the last three qubits with $\ketbra{W}{W}$. For example, for $k=5$, setting our instance to be the clause $\ketbra{\phi^+}{\phi^+}\otimes\ketbra{W}{W}$ yields an upper bound of $(1/2)(4/9)=2/9<1/4=d^{-\lfloor\frac{k}{2}\rfloor}$ (where we again use the value $4/9$ for $\ket{W}$ from the previous paragraph). For general odd $k>1$, this improved bound generalizes to $2^{\frac{-k+7}{2}}/9$.

\section{Optimizing over the set of separable quantum states}\label{scn:sepopt}
Section~\ref{scn:prodratio} showed that there always \emph{exists} a product state assignment achieving a certain non-trivial approximation ratio. In this section, we show how to efficiently \emph{find} such a product state. Our main theorem of this section is the following (Thm.~\ref{thm:ptas}), from which Thm.~\ref{thm:ptasdense} follows easily (see discussion at end of Sec.~\ref{sscn:finalptas}). 

\begin{theorem}\label{thm:ptas}
    Let $H$ be an instance of \klh~acting on $n$ qudits, and let $\optprod$ denote the optimum value of $\trace(H\rho)$ over all \emph{product} states $\rho\in D((\complex^d)^{\otimes n})$. Then, for any fixed $\epsilon>0$, there exists a polynomial time (deterministic) algorithm which outputs $\rho_1\otimes\cdots\otimes\rho_n\in D((\complex^d)^{\otimes n})$ such that
$
        \trace(H\rho_1\otimes\cdots\otimes\rho_n)\geq \optprod -\epsilon n^k.
$
\end{theorem}

We first outline our approach by generalizing the discussion in Sec.~\ref{scn:intro}, introducing tools and notation we will require along the way. The optimal value $\optprod$ over product state assignments for any \klh~instance can be expressed as the following program, denoted $P_1$:
\begin{equation}
      \optprod\hspace{2mm} = \hspace{2mm}\max \hspace{2mm}\sum_{i_1,\ldots,i_k}^n\trace(H_{i_1,\ldots,i_k}\rho_{i_1}\otimes\cdots\otimes\rho_{i_k})
    \hspace{2mm}\mbox{s.t.}\hspace{2mm} \rho_i\succeq 0\hspace{2mm} \mbox{and}\hspace{2mm}\trace(\rho_i)=1\hspace{2mm}\mbox{for } 1\leq i\leq n.\label{eqn:obj}
\end{equation}

\noindent As done in Eqn.~(\ref{eqn:introdecomp}), we now recursively decompose our objective function as a sequence of nested sums. Let $\set{\sigma_i}_{i=1}^d$ be a traceless, Hermitian orthogonal basis for the set of Hermitian operators acting on $\complex^d$, such that $\trace(\sigma_i\sigma_j)=2\delta_{ij}$ (for $\delta_{ij}$ the Kroenecker delta) (see, e.g.~\cite{K03}). Then, by rewriting each $H_{i_1,\ldots,i_k}$ in terms of $\set{\sigma_i}_{i=1}^d$, our objective function becomes
\begin{align}
    \sum_{i_k,\ldots,i_1}^n\trace\left[\left(\sum_{j_k,\ldots,j_1=1}^{d^2}r_{j_1,\ldots,j_k}^{i_1,\ldots,i_k}\sigma_{j_k}\otimes\cdots\otimes\sigma_{j_1}\right)\rho_{i_k}\otimes\cdots\otimes\rho_{i_1}\right] =\hspace{45mm}\nonumber\\
    \sum_{i_k,j_k}\trace(\sigma_{j_k}\rho_{i_k})\left[\sum_{i_{k-1},j_{k-1}}\trace(\sigma_{j_{k-1}}\rho_{i_{k-1}})\left[\cdots\left[\sum_{i_1}\trace\left(\left(\sum_{j_1}r_{j_1,\ldots,j_k}^{i_1,\ldots,i_k}\sigma_{j_1}\right)\rho_{i_1}\right)\right]\right]\right],\label{eqn:decomposition}
\end{align}
where each $\ve{r}^{i_1,\ldots,i_k}\in\reals^{d^2}$. We henceforth think of the objective function above as a ``degree-$k$ inner product'', i.e. as a sequence of $k$ nested sums involving inner products, in analogy to the degree-k polynomials of Ref.~\cite{AKK99}. In this sense, a degree-$1$ inner product would refer to only the innermost sums over $i_1$ and $j_1$, and a degree-$k$ inner product would denote the entire expression in Eqn.~(\ref{eqn:decomposition}). More formally, we denote\footnote{See the beginning of App.~\ref{app:A} for more elaborate notation used in the proofs of the claims of Sec.~\ref{scn:sepopt}.} a degree-$b$ inner product for $1\leq b \leq k$ using map $t_b:H(\complex^d)^{\times n}\mapsto\reals$, defined such that\footnote{Note that $t_b$ implicitly depends on parameters $i_{b+1},\ldots, i_k$ and $j_{b+1},\ldots , j_k$.} $t_b(\rho_1,\ldots,\rho_n):=    \sum_{i_b,j_b}\trace(\sigma_{j_b}\rho_{i_b})\left[\cdots\left[\sum_{i_1}\trace\left(\left(\sum_{j_1}r_{j_1,\ldots,j_k}^{i_1,\ldots,i_k}\sigma_{j_1}\right)\rho_{i_1}\right)\right]\right]$.

Our approach is to ``linearize'' the objective function of $P_1$ using exhaustive sampling and recursion to estimate its degree-$(k-1)$ inner products. To do so, we will require the Sampling Lemma.

\begin{lemma}[Sampling Lemma~\cite{AKK99}]\label{l:sample}
    Let $(a_i)$ be a sequence of $n$ real numbers with $\abs{a_i}\leq M$ for all $i$, and let $f>0$. If we choose a multiset of $s=g\log n$ of the $a_i$ at random (with replacement), then their sum $q$ satisfies
$
        \sum_i a_i - nM\sqrt{\frac{f}{g}} \leq q\frac{n}{s}\leq \sum_i a_i + nM\sqrt{\frac{f}{g}}
$
    with probability at least $1-n^{-f}$.
\end{lemma}

\noindent The proof of Lemma~\ref{l:sample} follows from a simple application of the H\"{o}ffding bound~\cite{H64}. To use the Sampling Lemma in conjunction with exhaustive sampling, we will discretize the space of $1$-qudit density operators using a $\delta$-net $G\subseteq H(\complex^{d})$, such that for all $\rho\in D(\complex^{d})$, there exists $\sigma\in G$ such that $\fnorm{\rho-\sigma}\leq \delta$. We now show how to construct $G$.

To obtain $G$, we instead construct a $\delta$-net for a subset of $H(\complex^d)$ which \emph{contains} $D(\complex^d)$, namely the set\footnote{Note: A net over $\mathcal{A}(\complex^d)$ may allow non-positive assignments for a qudit. See Sec.~\ref{sscn:finalptas} for why this is of no consequence.}  $\mathcal{A}(\complex^d):=\set{A\in H(\complex^d)\mid \max_{i,j} \abs{A(i,j)}\leq 1}$. Creating a $\delta$-net over $\mathcal{A}(\complex^d)$ is simple: we cast a $(\delta/d)$-net over the unit disk for each of the complex $d(d-1)/2$ matrix entries above the diagonal, and likewise over $[-1,1]$ for the entries on the diagonal. Letting $m$ and $n$ denote the minimum number of points required to create such $(\delta/d)$-nets for each of the diagonal and off-diagonal entries, respectively, we have that $\abs{G}=m^{\frac{d(d-1)}{2}}n^d$. For example, simple nets of size $m\approx d/\delta$ and $n\approx d^2/\delta^2$ can be obtained by placing a 1D and 2D grid over $[-1,1]$ and the length $2$ square in the complex plane centered at $(0,0)$, respectively, implying $\abs{G}\in O(1)$ when $d\in O(1)$. To show that $G$ is indeed a $\delta$-net, we now bound the Frobenius\footnote{We use the Frobenius norm as it allows a simple analysis. It is straightforward, however, to switch to say the trace norm using the fact that $\fnorm{X}\leq \sqrt{d}\trnorm{X}$ for all $X\in\complex^d$.} distance between arbitrary $\rho\in D(\complex^d)$ and the closest $\tilde{\rho}\in G$. Specifically, let $A:=\rho-\tilde{\rho}$. Then:
\[
    \fnorm{A}=\sqrt{\trace(A^\dagger A)}=\sqrt{\sum_{ij}\abs{A(i,j)}^2}\leq\sqrt{\sum_{ij}(\delta/d)^2}=\frac{\delta}{d}(d)=\delta.
\]

Finally, we remark that our \emph{dense} assumption on \klh~instances is only necessary to convert the absolute error of Thm.~\ref{thm:ptas} to a relative one~\cite{GK10} (this conversion is detailed in Sec.~\ref{sscn:finalptas}). The remaining sections are organized as follows: In Sec.~\ref{sscn:recurseSample}, we show how to recursively estimate degree-$b$ inner products using the Sampling Lemma. We then use this estimation technique in Sec.~\ref{sscn:linearization} to linearize our optimization problem $P_1$. Sec.~\ref{sscn:finalptas} brings everything together by presenting and analyzing the complete approximation algorithm. To ease reading of the remaining sections, all technical proofs are found in App.~\ref{app:A}. Please see the beginning of App.~\ref{app:A} for definitions of the more elaborate notation used in these proofs.

\subsection{Estimating degree-$k$ inner products using the Sampling Lemma}\label{sscn:recurseSample}

Our recursive procedure, EVAL, for estimating a degree-$k$ inner product using the Sampling Lemma is stated as Alg.~\ref{alg:estimate}. There are two sources of error we must analyze: the Sampling Lemma, and our $\delta$-net over $\complex^d$. We claim that EVAL estimates the degree-$b$ inner product $t_b(\rho_1,\ldots,\rho_n)$ to within additive error $\pm \epsilon_{b} n^{b}$, where $\epsilon_{b}$ is defined as follows. Set $\Delta := \sqrt{2}d(1+\delta)$, for $\delta$ from our $\delta$-net. Then,
\begin{equation}
    \epsilon_{b} := \cc \left(\sqrt{\frac{f}{g}}+\delta\right)\left(\frac{\Delta^{b}-1}{\Delta-1}\right).\label{eqn:alg1error}
\end{equation}
The following lemma formalizes this claim. We adopt the convention of~\cite{AKK99} and let $x\in y\pm z$ denote $x\in [y,z]$. Alg.~\ref{alg:estimate} is our operator analogue of the algorithm \emph{Eval} in Section 3.3 of~\cite{AKK99}. 

\begin{figure}[t]
\noindent\rule{\linewidth}{0.3mm}
\begin{alg}\rm{EVAL( }$t_{b}$ , $S$ , $\set{\tilde{\rho}_i:i\in S}$\rm{ )}.\label{alg:estimate}
    \begin{itemize}
        \item Input:\hspace{3mm}(1) A degree-$b$ inner product $t_{b}:H(\complex^d)^{\times n}\mapsto\reals$ for $1\leq b\leq k$\\
        \mbox{\hspace{12mm}}(2) A subset $S\subseteq\set{1,\ldots,n}$ of size $\abs{S}=O(\log n)$\\
         \mbox{\hspace{12mm}}(3) Sample points $\set{\tilde{\rho}_i:i\in S}$ such that $\fnorm{\tilde{\rho}_i-\rho_i}\leq \delta$ for all $i\in S$
        \item Output: $x\in\reals$ such that $x\in t_b(\rho_1,\ldots,\rho_n) \pm \epsilon_b n^{b}$ (for $\epsilon_b$ defined in Eqn.~(\ref{eqn:alg1error})).
    \end{itemize}
    \begin{compactenum}
        \item For all $i\in S$ and $j= 1\ldots d^2$:
        \begin{compactenum}
            \item (Base Case) if $b=1$, set $e_{ij}=1$.
            \item (Recurse) else, set $e_{ij}$ = EVAL($t_{b-1}^{ij},S,\set{\tilde{\rho}_i:i\in S})$.
        \end{compactenum}
        \item Return $\frac{n}{\abs{S}}\sum_{i\in S}\left[\sum_{j=1}^{d^2}\trace(\sigma_{j}\tilde{\rho}_{i})e_{ij}\right]$.
    \end{compactenum}
\end{alg}
\noindent\rule{\linewidth}{0.3mm}
\end{figure}

\begin{lemma}\label{l:evalbound}
    Let $t_{k}:H(\complex^k)^{\times n}\mapsto\reals$ be defined using set $\set{H_{i_1,\ldots,i_k}}\subseteq H((\complex^d)^{\otimes k})$ (as in Eqn.~(\ref{eqn:decomposition})). Let $S\subseteq \set{1,\ldots , n}$ such that $\abs{S}=g\log n$ have its elements chosen uniformly at random with replacement. Let $\rho_1,\ldots ,\rho_n \in D(\complex^d)$ be some assignment on all $n$ qudits, and $\set{\tilde{\rho}_i:i\in S}$ a set of elements in our $\delta$-net such that $\fnorm{\rho_i-\tilde{\rho}_i}\leq \delta$ for all $i\in S$. Then, for $1\leq b\leq k$, with probability at least $1-d^{2b}n^{b-f}$, we have
$
        \operatorname{EVAL}(t_{b},S,\set{\tilde{\rho}_i:i\in S})\in t_b(\rho_1,\ldots,\rho_n) \pm \epsilon_{b} n^{b},
$
    where $\epsilon_{b}$ is defined as in Eqn.~(\ref{eqn:alg1error}).
\end{lemma}

\subsection{Linearizing our optimization problem}\label{sscn:linearization}

Our procedure, LINEARIZE, for ``linearizing'' the objective function of $P_1$ using EVAL from Sec.~\ref{sscn:recurseSample} is stated as Alg.~\ref{alg:linearize}. Alg.~\ref{alg:linearize} takes as input $P_1$ and a set of sample points $\set{\tilde{\rho_i}}$, and outputs a semidefinite program (SDP) which we shall henceforth refer to as $P_2$. We remark that LINEARIZE is our version of the procedure \emph{Linearize} in Sec.~3.4 of~\cite{AKK99}, extended to the setting of operators and a more complex error structure. Although LINEARIZE is presented as linearizing an objective function here, the same techniques straightforwardly apply in linearizing constraints involving high-degree inner products.

\begin{figure}[t]
\noindent\rule{\linewidth}{0.3mm}
\begin{alg}\rm{LINEARIZE( }$t_b$ , $\mathcal{N}$ , $S$, $\set{\tilde{\rho}_i:i\in S}$, $\epsilon$, $U$, $L$\rm{ )}.\label{alg:linearize}
    \begin{itemize}
        \item Input:\hspace{3mm}(1) A degree-$b$ inner product $t_{b}:H(\complex^d)^{\times n}\mapsto\reals$ for $1\leq b\leq k$.\\
               \mbox{\hspace{12mm}}(2) A set of linear constraints $\mathcal{N}$ (e.g. ``$\rho_i\succeq 0$'').\\
               \mbox{\hspace{12mm}}(3) A subset $S\subseteq\set{1,\ldots,n}$ of size $\abs{S}=O(\log n)$.\\
         \mbox{\hspace{12mm}}(4) Sample points $\set{\tilde{\rho}_i:i\in S}$ consistent with some feasible solution $(\rho_1,\ldots,\rho_n)$ for \mbox{\hspace{18mm}}$P_1$ such that $\fnorm{\tilde{\rho}_i-\rho_i}\leq \delta$ for all $i\in S$.  \\
               \mbox{\hspace{12mm}}(5) An error parameter $\epsilon>0$.\\
                \mbox{\hspace{12mm}}(6) (Optional) upper and lower bounds $U,L\in\reals$. If $U$ and $L$ are not provided, we\\
                \mbox{\hspace{19mm}}assume $U,L=\infty$.
        \item Output: (1) (Optional) A linear objective function $f:(L(\complex^d))^{ \times n}\rightarrow \reals$.\\
                \mbox{\hspace{13mm}}(2) An updated set of linear constraints, $\mathcal{N}$.
    \end{itemize}
    \begin{compactenum}
        \item (Base case) If $b=1$, then
        \begin{compactenum}
            \item (Trivial: Initial objective function was linear) If $U=L=\infty$, return [$t_b$, $\mathcal{N}$].
            \item (Reached bottom of recursion) Else, return [$\mathcal{N}\cup \set{``L\leq t_{b}(\rho_1,\ldots,\rho_n)\leq U"}$].
        \end{compactenum}
        \item (Recursive case) For $i=1\ldots n$ and $j=1\ldots d^2$ do
        \begin{compactenum}
            \item Set $e_{ij}:=\operatorname{EVAL}(t_{b-1}^{ij},S,\set{\tilde{\rho}_i:i\in S})$.
            \item Set $\epsilon^\prime := \epsilon - \cc\left(\sqrt{\frac{f}{g}}+\delta\right)\Delta^{b-1}$, for $\Delta$ defined in Eqn.~(\ref{eqn:alg1error}).
            \item Set $l_{ij}:=e_{ij}-\epsilon^\prime n^{b-1}$ and $u_{ij}:=e_{ij}+\epsilon^\prime n^{b-1}$.
            \item Call LINEARIZE($t_{b-1}^{ij},\mathcal{N},S,\set{\tilde{\rho}_i:i\in S}, \epsilon^\prime,u_{ij},l_{ij}$).
        \end{compactenum}
        \item (a) (Entire computation done) If $U=L=\infty$, return $\left[\sum_{ij}\trace(\sigma_{j}{\rho}_{i})e_{ij}, \mathcal{N}\right]$.\\
        (b) (Recursive call done) Else, return $\left[\mathcal{N}\cup \set{``L-\epsilon^\prime d^2n^{b}\leq \sum_{ij}\trace(\sigma_{j}{\rho}_{i})e_{ij}\leq U+\epsilon^\prime d^2n^{b}"}\right]$.
    \end{compactenum}
\end{alg}
\noindent\rule{\linewidth}{0.3mm}
\end{figure}


To prove correctness of our final approximation algorithm, we require the following two important lemmas regarding $P_2$. The first shows that any feasible solution $(\rho_1,\ldots,\rho_n)$ for $P_1$ consistent with the sample set $\set{\tilde{\rho}_i:i\in S}$ fed into LINEARIZE is also a feasible solution for $P_2$ with high probability.

\begin{lemma}\label{l:feasible}
    Let $t_{k}$, assignment $(\rho_1,\ldots,\rho_n)$, $S$, and $\set{\tilde{\rho}_i:i\in S}$ be defined as in Lem.~\ref{l:evalbound}. Then, for any $f,g>0$, calling LINEARIZE with parameters $t_{k}$, $\set{\tilde{\rho}_i:i\in S}$, and $\epsilon=\epsilon_k$ (for $\epsilon_k$ defined in Eqn.~(\ref{eqn:alg1error})) yields an SDP $P_2$ for which the assignment $\set{\rho_1,\ldots,\rho_n}$ is feasible with probability at least $1-d^{2k}n^{k-f}$.
\end{lemma}

The second lemma is a bound on how far the optimal solution of $P_2$ is from the optimal solution for $P_1$. We adopt the convention of~\cite{AKK99} and write $[x,y]\pm z$ to denote interval $[x-z,y+z]$.

\begin{lemma}\label{l:guarantee}
    Let $\optprod$ be the optimal value for $P_1$ with corresponding assignment $\rho^{\optprod}:=(\rho^{\operatorname{opt}}_1,\ldots,\rho^{\operatorname{opt}}_n)$, and let $\set{\tilde{\rho}_i:i\in S}$ be such that $\fnorm{\tilde{\rho}_i-\rho^{\operatorname{opt}}_i}\leq \delta$ for all $i\in S$ for some $S\subseteq\set{1,\ldots,n}$. Let $P_2$ denote the SDP obtained by calling LINEARIZE with $S$, and denote by $\epsilon_m$ for $1\leq m \leq k$ the error parameter passed with map $t_m$ into a (possibly recursive) call to LINEARIZE. Then, letting $\optt$ denote the optimal value of $P_2$, we have with probability at least
     $1-d^{2k}n^{k-f}$ (for parameters set as in Lem.~\ref{l:feasible}) that
 $       \optt \in \optprod \pm d(d+\sqrt{2})\left[\sum_{m=1}^{k-1}(\sqrt{2}d)^{k-1-m}\epsilon_m\right]n^{k}.$
\end{lemma}

\subsection{The final algorithm}\label{sscn:finalptas}

We finally present our approximation algorithm, APPROXIMATE (Alg.~\ref{alg:final}), in its entirety, which exploits our ability to linearize $P_1$ using LINEARIZE (Alg.~\ref{alg:linearize}). This proves Thm~\ref{thm:ptas}, which in turn implies Thm.~\ref{thm:ptasdense}. We first clarify a few points about APPROXIMATE, then analyze its runtime, and follow with further discussion, including the algorithm's derandomization and a proof that dense \klh~remains QMA-hard.

\begin{figure}[t]
\noindent\rule{\linewidth}{0.3mm}
\begin{alg}{\rm APPROXIMATE}($H$ , $\epsilon${\rm )}.\label{alg:final}
    \begin{itemize}
        \item Input: (1) A $k$-local Hamiltonian $H=\sum_{i_1, \ldots ,i_k} H_{i_1,\ldots,i_k}$ for each $H_{i_1,\ldots,i_k}\in H((\complex^d)^k)$.\\
               \mbox{\hspace{11mm}}(2) An error parameter $\epsilon>0$.
        \item Output: A product assignment $\rho_1\otimes\cdots\otimes\rho_n$ that with probability at least $1/2$, has value at least\\
              \mbox{\hspace{14mm}}$\optprod-\epsilon n^k$, for $\optprod$ the optimal value for $H$ over all product state assignments.
    \end{itemize}
    \begin{compactenum}
        \item Set $\epssdp:=\eps/10$.
        \item Define $h:\reals\rightarrow\reals$ such that for any error parameter $\epsilon$ input to LINEARIZE, $h(\eps)n^k$ is the absolute value of the bound on additive error given by Lem.~\ref{l:guarantee}. Then, define $\epsilon^\prime$ implicitly so that $h(\epsilon^\prime) + \epssdp=\epsilon$ holds.
        \item Define constant $f$ such that $1-d^{2k}n^{k-f}>1/2$.
        \item Define constants $g$ and $\delta$ implicitly so that $\epsilon^\prime=\cc\left(\sqrt{\frac{f}{g}}+\delta\right)\left(\frac{\Delta^{k}-1}{\Delta-1}\right)$, for $\Delta$ defined in Eqn.~(\ref{eqn:alg1error}).
        \item Choose $g\log n$ indices $S\subseteq\set{1,\ldots,n}$ independently and uniformly at random.
        \item For each possible assignment $i$ from our $\delta$-net to the qudits in $S$:
        \begin{compactenum}
            \item Call LINEARIZE$(t_k,\set{P_1\mbox{'s constraints}},S,i, \epsilon^\prime)$ to obtain SDP $P_2^i$.
            \item Let $\alpha_i$ denote the value of $P_1$ obtained by substituting in the optimal solution of $P_2^i$.
        \end{compactenum}
        \item Return the assignment corresponding to the maximum over all $\alpha_i$.
    \end{compactenum}
\end{alg}
\noindent\rule{\linewidth}{0.3mm}
\end{figure}

We begin by explaining the rationale behind the constants in Alg.~\ref{alg:final}. The constant $\epssdp$ is the additive error incurred when solving an SDP~\cite{GLS93}. We choose $\epsilon^\prime$ so that after running LINEARIZE and solving $P_2^i$, the total additive error is at most $\epsilon$, as desired. We choose $f$ to ensure the probability of success is at least $1/2$. Finally, we set $g$ large enough and $\delta$ (for our $\delta$-net) small enough to ensure that $\epsilon^\prime$ matches the error bounds for EVAL in Lem.~\ref{l:evalbound}.

We now analyze the runtime of Alg.~\ref{alg:final}. Let $\abs{G}$ denote the size of our $\delta$-net $G$ for a qudit. Then, for each of the $\abs{G}^{g\log n}$ iterations of line 6, we first take $O(n^{k-1})$ time to run LINEARIZE, outputting $O(n^{k-1})$ new linear constraints (seen via a simple inductive argument). We then solve SDP $P_2^i$, which can be done in time polynomial in $n$ and $\log(1/\epssdp)$ using the ellipsoid method~\cite{GLS93} (see, e.g.,~\cite{W09}). Let $r(n,\epssdp)$ denote the maximum runtime required to solve any of the $P_2^i$. Then, the overall runtime for Alg.~\ref{alg:final} is $O(n^{g\log\abs{G}} (n^{k-1}+r(n,\epssdp)))$, which is polynomial in $n$ for $\epsilon,d,k\in O(1)$ (recall from Sec.~\ref{scn:sepopt} that $\abs{G}\in O((\frac{d}{\delta})^d)$, and that $\delta$ and $g$ are constant in our setting). Note that, due to the implicit dependence of $g$ on $\epsilon$, this runtime scales at least exponentially with varying $\epsilon$.

Before moving to further discussion, we make two remarks. First, one can convert the output of Alg.~\ref{alg:final} to a \emph{pure} state with the same guarantee by adapting the standard classical \emph{method of conditional expectations}~\cite{V01}. To demonstrate, suppose $\set{\rho_i}$ is output by Alg.~\ref{alg:final}. Then, set $\rho_1^\prime$ to be the eigenvector $\ketbra{\psi_j}{\psi_j}$ of $\rho_1$ for which the assignment $\ketbra{\psi_j}{\psi_j}\otimes\rho_2\otimes\cdots\otimes\rho_n$ performs best\footnote{If the spectrum of $\rho_i$ is degenerate, begin by fixing an arbitrary choice of spectral decomposition for $\rho_i$.} for $P_1$. Let our new assignment be $\rho_1^\prime\otimes\rho_2\otimes\cdots\otimes\rho_n$. Now repeat for each $\rho_i$ for $2\leq i \leq n$. The final state $\rho_1^\prime\otimes\cdots\otimes\rho_n^\prime$ is pure, and by convexity is guaranteed to perform as well as $\rho_1\otimes\cdots\otimes\rho_n$.

Second, recall from Sec.~\ref{scn:sepopt} that we constructed a $\delta$-net over a space larger than $D(\complex^d)$, allowing possibly non-positive assignments for a qudit. We now see that this is of no consequence, since regardless of which samples (positive or not) we use to derive our estimates with the Sampling Lemma, any feasible solution to $P_2^i$ in Alg.~\ref{alg:final} is a valid assignment for $P_1$. Moreover, we know that for each optimal $\rho_i$ for $P_1$, there must be \emph{some} operator (positive or not) within distance $\delta$ in our net, ensuring our estimates obtained using the Sampling Lemma are within our error bounds.

\paragraph{Converting the absolute error of Algorithm~\ref{alg:final} into relative error.}
To convert the absolute error $\pm\epsilon n^k$ of Alg.~\ref{alg:final} into a \emph{relative} error of $1-\epsilon^\prime$ for any $\epsilon^\prime$, define constant $c$ such that $c n^k$ is the value obtained for a \klh~instance by choosing the maximally mixed assignment $I/d^n$ (analogous to a classical random assignment). Since $I/d^n$ can be written as a mixture of computational basis states, we have $\optprod \geq c n^k$. It follows that by setting $\epsilon= c\epsilon^\prime$, Alg.~\ref{alg:final} returns an assignment with value at least $\optprod-c\epsilon^\prime n^k\geq \optprod-\epsilon^\prime\optprod\geq \optprod(1-\epsilon^\prime)$, as desired.
\paragraph{Derandomizing Algorithm~\ref{alg:final}.}

The source of randomness in our algorithm is Lem.~\ref{l:sample}. By a standard argument in~\cite{AKK99} (see also~\cite{BR94,BGG93}), this randomness can be eliminated with only polynomial overhead. Specifically, we replace the random selection of $g\log n$ indices in the Sampling Lemma with the set of indices encountered on a random walk of length $O(g\log n)$ along a constant degree expander~\cite{G93}. Since the expander has constant degree, we can efficiently deterministically iterate through all $n^{O(g)}$ such walks, and since such a walk works with probability $1/n^{O(1)}$, at least one walk will work for all $\poly(n)$ sampling experiments we wish to run.

\paragraph{QMA-hardness of dense \klh.}
It is easy to see that (exact) MAX-$2$-local Hamiltonian remains QMA-hard for dense instances (a similar statement holds for MAX-$2$-SAT~\cite{AKK99}). For any MAX-$2$-local Hamiltonian instance with optimal value $\OPT$, we simply add $n$ qudits, between any two of which we place the constraint $\ketbra{00}{00}$ (no constraints are necessary between old and new qudits). Then, the new Hamiltonian has optimal value $\OPT+{n\choose{2}}$, making it dense, and the ability to solve this new instance implies the ability to solve the original one. The argument extends straightforwardly to \klh~for $k>2$.

%
%

\section{Acknowledgements}
We thank Jamie Sikora and Sarvagya Upadhyay for helpful feedback on an earlier version of this draft, and Yi-Kai Liu for interesting discussions. We wish to especially thank Oded Regev for many helpful comments and suggestions, and Richard Cleve for bringing our attention to the method of conditional expectations, and for stimulating discussions and support.

\bibliography{approx}

\appendix
\section{Technical proofs for Section \ref{scn:sepopt}}\label{app:A}

\paragraph{Expanded Notation.} We now expand on our previous notation for analyzing Eqn.~(\ref{eqn:decomposition}) in order to facilitate proofs of the claims in Sec.~\ref{scn:sepopt}. First, to recursively analyze a clause $H_{i_1,\ldots,i_k}\subseteq H((\complex^{d})^{\otimes k})$, let $H_b\in H((\complex^{d})^{\otimes b})$ for any $1\leq b \leq k$ denote the action of $H_{i_1,\ldots,i_k}$ restricted to the first $b$ of its $k$ target qubits, i.e. $H_b:=\sum^{d^2}_{j_b,\ldots,j_1}r_{j_1,\ldots,j_k}^{i_1,\ldots,i_k}\sigma_{j_{b}}\otimes\cdots\otimes\sigma_{j_{1}}$. For example, $H_1=\sum^{d^2}_{j_1}r_{j_1,\ldots,j_k}^{i_1,\ldots,i_k}\sigma_{j_{1}}$ and  $H_k=H_{i_1,\ldots,i_k}$. Note that $H_b$ implicitly depends on variables ${i_1,\ldots,i_k,j_{b+1},\ldots, j_{k}}$, but to reduce clutter, our notation does not explicitly denote this dependence unless necessary. Next, to recursively analyze a degree-$a$ inner product, we define $t_{a,b}:H(\complex^d)^{\times n}\mapsto\reals$ for any $0\leq a \leq k$ and $1\leq b \leq k$ such that $t_{a,b}(\rho_1,\ldots,\rho_n):=\sum^n_{i_{a},\ldots,i_1}\trace\left(H_b^{i_1,\ldots,i_k}\rho_{i_{b}}\otimes\cdots\otimes\rho_{i_1}\right)$ (where setting $a=0$ eliminates the sum over indices $i$). For example, $t_{k,k}$ is our full ``degree-$k$'' objective function in Eqn.~(\ref{eqn:obj}), and more generally, $t_{b,b}$ is the degree-b inner product in Eqn.~(\ref{eqn:decomposition}). Allowing different values for $a$ and $b$ greatly eases our technical analysis. We use the shorthand $t_b$ to denote $t_{b,b}$, and again only explicitly denote the dependence of $t_{a,b}$ on parameters $i_{a+1},\ldots, i_k$ and $j_{b+1},\ldots , j_k$ when necessary.

\begin{lemma}\label{l:ubound}
    Let $\set{\rho_i}_{i=1}^n\subseteq H(\complex^d)$. For any \klh~instance  $\set{H_{i_1,\ldots,i_k}}\subseteq H(\complex^{d^k})$ with decomposition for the $H_{i_1,\ldots,i_k}$ as given in Eqn.~(\ref{eqn:decomposition}), we have for any $0\leq a \leq k$ and $1\leq b \leq k$ that $\abs{t_{a,b}(\rho_1,\ldots,\rho_n)}\leq \left(\max_{i_{b},\ldots,i_1}\fnorm{\rho_{i_{b}}}\cdots\fnorm{\rho_{i_{1}}}\right)\cc n^{a}$.
\end{lemma}
\begin{proof}[\textbf{Proof of Lem.~\ref{l:ubound}}]
     By the triangle inequality and the H\"{o}lder inequality for Schatten $p$-norms, we have
    \begin{eqnarray*}
        \abs{t_{a,b}}=\abs{\sum^n_{i_{a},\ldots,i_1}\trace\left(H_b\rho_{i_{b}}\otimes\cdots\otimes\rho_{i_1}\right)}&\leq& \sum^n_{i_{a},\ldots,i_1}\fnorm{H_b}\fnorm{\rho_{i_{b}}\otimes\cdots\otimes\rho_{i_1}}\\
        &\leq& \left(\max_{i_{b},\ldots,i_1}\fnorm{\rho_{i_{b}}}\cdots\fnorm{\rho_{i_{1}}}\right) \sum^n_{i_{a},\ldots,i_1}\fnorm{H_b},
    \end{eqnarray*}
    where we have used the fact that $\fnorm{A\otimes B}=\fnorm{A}\fnorm{B}$ for all $A,B\in L(\complex^d)$. If we can now show that $\fnorm{H_b}\leq \fnorm{H_k}$ for all $1\leq b \leq k$, then we would be done since we would have $\sum^n_{i_{a},\ldots,i_1}\fnorm{H_b}\leq \fnorm{H_{k}}n^{a}\leq \cc n^a$, where $\fnorm{H_{k}}\leq \cc$ since $\snorm{H_k}\leq 1$ by definition. Indeed, we claim that for any fixed $1\leq b\leq k$, we have $\fnorm{H_b}\leq 2^{\frac{b}{2}}\fnorm{H_{k}}$. To see this, note by straightforward expansion of the Frobenius norm and the fact that $\trace(\sigma_i\sigma_j)=2\delta_{ij}$ that
    \[
        \fnorm{H_b} = \sqrt{\trace(H_b^2)}= 2^{\frac{b}{2}}\sqrt{\sum_{j_{b},\ldots,j_k}(r_{j_1,\ldots,j_k}^{i_1,\ldots,i_k})^2}\leq2^{\frac{b}{2}}\enorm{\ve{r}^{i_1,\ldots,i_k}}=2^{\frac{b-k}{2}}\left(2^{\frac{k}{2}}\enorm{\ve{r}^{i_1,\ldots,i_k}}\right),
    \]
    where $\ve{r}^{i_1,\ldots,i_k}$ is the coordinate vector of $H_{i_1,\ldots,i_k}$ from Eqn.~(\ref{eqn:decomposition}). By the second equality in the chain above, we see that in fact $\fnorm{H_{k}}=2^{\frac{k}{2}}\enorm{\ve{r}^{i_1,\ldots,i_k}}$, completing the proof of our claim.
\end{proof}

\begin{proof}[\textbf{Proof of Lem.~\ref{l:evalbound}}]
    We first derive the error bound of $\epsilon_b$, and subsequently prove the probability bound. We follow~\cite{AKK99}, and proceed by induction on $b$. For the base case $b=1$, $\operatorname{EVAL}(H_{1},S,\set{\tilde{\rho}_i:i\in S})$ attempts to estimate
    \[
        t_1(\rho_1,\ldots,\rho_n)=\sum_{i_1}\left[\sum_{j_1}r_{j_1,\ldots,j_k}^{i_1,\ldots,i_k}\trace(\sigma_{j_1}\rho_{i_1})\right]
    \]
    using our flawed sample points $\set{\tilde{\rho}_i:i\in S}$. To analyze the error of its output, assume first that our sample points are exact, i.e. $\tilde{\rho}_i=\rho_i$ for all $i\in S$. Then, by setting ``$a_i$'' in Lem.~\ref{l:sample} to $t_{0,1}^{i_1}$ for $i=i_1$, and by using Lem.~\ref{l:ubound} with parameters $a=0$ and $b=1$ to obtain upper bound $M=\cc$, we have by the Sampling Lemma that (with probability at least $1-n^{-f}$)
    \begin{equation}\label{eqn:basecaseerror}
        \frac{n}{\abs{S}}\sum_{i_1\in S}\left[\sum_{j_1}r_{j_1,\ldots,j_k}^{i_1,\ldots,i_k}\trace(\sigma_{j_1}\rho_{i_1})\right]\in t_1(\rho_1,\ldots,\rho_n)\pm \cc\sqrt{\frac{f}{g}}n.
    \end{equation}
    This bound holds if we sum over exact sample points. If we instead sum over flawed sample points $\set{\tilde{\rho}_i:i\in S}$, the additional error is bounded by $\frac{n}{\abs{S}}$ times
    \begin{equation}\label{eqn:sampleerror}
        \abs{\sum_{i_1\in S}\left[\sum_{j_1}r_{j_1,\ldots,j_k}^{i_1,\ldots,i_k}\trace(\sigma_{j_1}(\rho_{i_1}-\tilde{\rho}_{i_1}))\right]}\leq \sum_{i_1\in S}\abs{\sum_{j_1}r_{j_1,\ldots,j_k}^{i_1,\ldots,i_k}\trace(\sigma_{j_1}(\rho_{i_1}-\tilde{\rho}_{i_1}))}\leq \sum_{i_1\in S}(\fnorm{\rho_{i_1}-\tilde{\rho}_{i_1}}\cc)\leq \cc\delta n,
    \end{equation}
    where the second inequality uses Lem.~\ref{l:ubound} with parameters $a=0$ and $b=1$ and the promise of our $\delta$-net. We conclude for the base case that
    \[
        \operatorname{EVAL}(H_{1},S,\set{\tilde{\rho}_i:i\in S})=\frac{n}{\abs{S}}\sum_{i_1\in S}\left[\sum_{j_1}r_{j_1,\ldots,j_k}^{i_1,\ldots,i_k}\trace(\sigma_{j_1}\tilde{\rho}_{i_1})\right]\in t_1(\rho_1,\ldots,\rho_n)\pm \cc\left(\sqrt{\frac{f}{g}}+\delta\right)n,
    \]
    as desired.

    Assume now that the inductive hypothesis holds for $1\leq m \leq b-1$. We prove the claim for $m=b$. To do so, suppose first that the recursive calls on line 1(b) of Alg.~\ref{alg:estimate} return the \emph{exact} values of $t_{b-1}^{ij}(\rho_1,\ldots,\rho_n)$, and that we have exact samples $\set{{\rho}_i:i\in S}$. Then, since by calling Lem.~\ref{l:ubound} with $a=b-1$ we have $\abs{\sum_{j}\trace(\sigma_{j}\rho_{i})t_{b-1}^{ij}(\rho_1,\ldots,\rho_n)}\leq \cc n^{b-1}$, it follows by the Sampling Lemma that
    \begin{equation}
        \frac{n}{\abs{S}}\sum_{i\in S}\left[\sum_{j}\trace(\sigma_{j}\rho_{i})t_{b-1}^{ij}(\rho_1,\ldots,\rho_n)\right]\in         \sum_{i=1}^n\left[\sum_{j}\trace(\sigma_{j}\rho_{i})t_{b-1}^{ij}(\rho_1,\ldots,\rho_n)\right]\pm \cc\sqrt{\frac{f}{g}}n^{b}.\label{eqn:reccase1}
    \end{equation}
    To first adjust for using flawed samples, observe that an analogous calculation to Eqn.~(\ref{eqn:sampleerror}) yields $\abs{\frac{n}{\abs{S}}\sum_{i\in S}\left[\sum_{j}\trace(\sigma_{j}(\rho_{i}-\tilde{\rho}_{i}))\right]}\leq \cc\delta n^{b}$,
    where we have called Lem.~\ref{l:ubound} with $a=b-1$. Thus, using flawed samples, the output of Alg.~\ref{alg:estimate} satisfies
    \begin{equation}
        \frac{n}{\abs{S}}\sum_{i\in S}\left[\sum_{j}\trace(\sigma_{j}\tilde{\rho}_{i})t_{b-1}^{ij}\right]\in         \sum_{i=1}^n\left[\sum_{j}\trace(\sigma_{j}\rho_{i})t_{b-1}^{ij}\right]\pm \cc\left(\sqrt{\frac{f}{g}}+\delta\right)n^{b}.\label{eqn:exact}
    \end{equation}
    To next drop the assumption that our estimates $e_{ij}$ on line 1(b) are exact, apply the induction hypothesis to conclude that $e_{ij}\in t_{b-1}^{ij}(\rho_1,\ldots,\rho_n) \pm \epsilon_{b-1}n^{b-1}$. Then,
    \begin{eqnarray}
        \frac{n}{\abs{S}}\sum_{i\in S}\left[\sum_{j}\trace(\sigma_{j}\tilde{\rho}_{i})e_{ij}\right]&\in&
        \frac{n}{\abs{S}}\sum_{i\in S}\left[\sum_{j}\trace(\sigma_{j}\tilde{\rho}_{ij})\left(t_{b-1}^{ij} \pm \epsilon_{b-1}n^{b-1}\right)\right]\nonumber\\
        &\subseteq&
        \frac{n}{\abs{S}}\sum_{i\in S}\left[\sum_{j}\trace(\sigma_{j}\tilde{\rho}_{i})t_{b-1}^{ij}\right] \pm \frac{\epsilon_{b-1}n^{b}}{\abs{S}}\sum_{i\in S}\left[\sum_{j=1}^{d^2}\trace(\sigma_{j}\tilde{\rho}_{i}) \right]\nonumber\\
        &\subseteq&
        \frac{n}{\abs{S}}\sum_{i\in S}\left[\sum_{j}\trace(\sigma_{j}\tilde{\rho}_{i})t_{b-1}^{ij}\right] \pm \epsilon_{b-1}\sqrt{2}d(1+\delta)n^{b}\label{eqn:recurse2},
    \end{eqnarray}
    where the last statement follows since
    \begin{equation}
        \abs{\sum_{j=1}^{d^2}\trace(\sigma_{j}\tilde{\rho}_{i})}=\abs{\sum_{j=1}^{d^2}\trace\left(\sigma_{j}\left(\sum_{m=1}^{d^2} \tilde{r}_m \sigma_m\right)\right)}\leq 2\sum_{m=1}^{d^2} \abs{\tilde{r}_m}\leq 2d\enorm{\ve{\tilde{r}}}\leq\sqrt{2}d(1+\delta),\label{eqn:yetanotherequation}
    \end{equation}
    where $\ve{\tilde{r}}$ denotes the coordinate vector of $\tilde{\rho}_{i}$ with respect to basis $\set{\sigma_m}$, and we have used the facts that $\trace(\sigma_i\sigma_j)=2\delta_{ij}$, that $\onorm{\ve{x}}\leq \sqrt{d}\enorm{\ve{x}}$ for $\ve{x}\in\complex^d$, that $\fnorm{\tilde{\rho}_{i}}=\sqrt{2}\enorm{\ve{\tilde{r}}}$ for any $\tilde{\rho}_{i}\in H(\complex^d)$, and that $\fnorm{\tilde{\rho}_{i}}\leq 1+\delta$ (which follows from our $\delta$-net and the triangle inequality). Thus, recalling that $\Delta=\sqrt{2}d(1+\delta)$ and substituting Eqn.~(\ref{eqn:exact}) into Eqn.~(\ref{eqn:recurse2}), we have that
    \[
        \frac{n}{\abs{S}}\sum_{i\in S}\left[\sum_{j}\trace(\sigma_{j}\tilde{\rho}_{i})e_{ij}\right]\in
        t_b(\rho_1,\ldots,\rho_n)\pm \left[\cc\left(\sqrt{\frac{f}{g}}+\delta\right)+\epsilon_{b-1}\Delta\right]n^{b}.
    \]
    We hence have the recurrence relation $\epsilon_b\leq \cc\left(\sqrt{\frac{f}{g}}+\delta\right)+\epsilon_{b-1}\Delta$, which when unrolled yields
    \[
        \epsilon_b\leq \cc\left(\sqrt{\frac{f}{g}}+\delta\right)\sum_{m=0}^{b-1}\Delta^m=\cc\left(\sqrt{\frac{f}{g}}+\delta\right)\left(\frac{\Delta^{b}-1}{\Delta-1}\right),
    \]
    as desired. This concludes the proof of the error bound.

    To prove the probability bound, we instead prove the stronger bound of $1-\left(\sum_{m=0}^{b-1}d^{2m} n^m\right)n^{-f}$ by induction on $b$. The base case $b=1$ follows directly from our application of the Sampling Lemma in Eqn.~(\ref{eqn:basecaseerror}). For the inductive step, define for brevity of notation $\gamma:=d^2 n$, and apply the induction hypothesis to line 1(b) of Alg.~\ref{alg:estimate} to conclude that each of the $\gamma$ calls to EVAL fails will probability at most $(\sum_{m=0}^{b-2}\gamma^m)n^{-f}$. Then, by the union bound, the probability that at least one call fails is at most $(\sum_{m=1}^{b-1}\gamma^m)n^{-f}$. Similarly, since our application of the Sampling Lemma in line 2 of Alg.~\ref{alg:estimate} fails with probability at most $n^{-f}$, we arrive at our claimed stronger bound of $1-\left(\sum_{m=0}^{b-1}\gamma^m\right)n^{-f}$, as desired.
\end{proof}

\begin{proof}[\textbf{Proof of Lem.~\ref{l:feasible}}]
    We begin by observing that if one sets $\epsilon=\epsilon_k$, then the value of $\epsilon^\prime$ in line 2(b) of Alg.~\ref{alg:linearize} is precisely $\epsilon_{k-1}$, and more generally, the $\epsilon$ passed into the recursive call of line 2(e) on $t_b$ for any $1\leq b\leq k$ is $\epsilon_b$. Now, focus on some recursive call on $t_b$ for $b>1$ (the case of $b=1$ is straightforward by Lem.~\ref{l:evalbound}). If the estimates $e_{ij}$ in line 2(a) succeed, then by Lem.~\ref{l:evalbound}, we know that $e_{ij}\in t_{b-1}^{ij}(\rho_1,\ldots,\rho_n)\pm \epsilon_{b-1}n^{b-1}$, implying $t_{b-1}^{ij}(\rho_1,\ldots,\rho_n)\in[l_{ij},u_{ij}]$. Now, $l_{ij}$ and $u_{ij}$ are only incorporated into linear constraints in recursive calls on $t_{b-1}^{ij}$, yielding constraints of the form
    \begin{equation}
        l_{i_bj_b}-\epsilon_{b-2} d^2n^{b-1}\leq \sum_{i_{b-1},j_{b-1}}\trace(\sigma_{j_{b-1}}{\rho}_{i_{b-1}})e_{i_{b-1}j_{b-1}}\leq u_{i_bj_b}+\epsilon_{b-2} d^2n^{b-1}.\label{eqn:lemma5_1}
    \end{equation}
    But $\set{\rho_1,\ldots,\rho_n}$ must now satisfy this constraint, since recall
    \[
        t_{b-1}(\rho_1,\ldots,\rho_n)=\sum_{i_{b-1},j_{b-1}}\trace(\sigma_{j_{b-1}}{\rho}_{i_{b-1}})t_{b-2}^{i_{b-1}j_{b-1}}(\rho_1,\ldots,\rho_n),
    \]
    and there are $d^2n$ terms $e_{i_{b-1}j_{b-1}}$ in Eqn.~(\ref{eqn:lemma5_1}) each yielding an additional error of at most $\epsilon_{b-2}n^{b-2}$ (assuming EVAL succeeded on $t_{b-2}^{i_{b-1}j_{b-1}}$ in line 2(a)) above and beyond the bounds $t_{b-1}^{ij}(\rho_1,\ldots,\rho_n)\in[l_{ij},u_{ij}]$ we established above.

    We conclude that if, for \emph{all} $b$, $i$, and $j$, EVAL succeeds in producing estimates $e_{i_b}^{ij}$, then $\set{\rho_1,\ldots,\rho_n}$ is a feasible solution for $P_2$, as desired. The probability of this happening is, by the proof of Lem.~\ref{l:evalbound}, at least $1-d^{2k}n^{k-f}$, since EVAL recursively estimates precisely the same terms during its execution\footnote{This holds even though on line 1 of Alg.~\ref{alg:estimate}, we only estimate $d^2\abs{S}$ of the terms $e_{ij}$ (i.e. EVAL does not actually estimate \emph{all} terms in the recursive decomposition of $t_k$, as it does not need to) --- this is because in our analysis of the probability bound for Alg.~\ref{alg:estimate}, we actually produced a looser bound by assuming all $n$ terms $e_{ij}$ are estimated.}.
\end{proof}

\begin{proof}[\textbf{Proof of Lem.~\ref{l:guarantee}}]
    We begin by proving that for any recursive call to LINEARIZE on $t_b$ with valid upper and lower bounds $U$ and $L$ (i.e. $U,L\neq\infty$), respectively, we have for \emph{any} feasible solution $(\rho_1,\ldots,\rho_n)$ to $P_2$ that
    \begin{equation}\label{eqn:miniclaim}
        t_b(\rho_1,\ldots,\rho_n) \in [L,U] \pm d(d+\sqrt{2})\left[\sum_{m=1}^{b-1}(\sqrt{2}d)^{b-1-m}\epsilon_m\right]n^{b}.
    \end{equation}

    We prove this by induction on $b$, following~\cite{AKK99}. For base case $b=1$, the claim is trivial by line 1(b) of the algorithm. Now, assume by induction hypothesis that
    \[
        t_{b-1}^{ij}(\rho_1,\ldots,\rho_n)\in [l_{ij},u_{ij}]\pm d(d+\sqrt{2})\left[\sum_{m=1}^{b-2}(\sqrt{2}d)^{b-2-m}\epsilon_m\right]n^{b-1}.
    \]
    By substituting the values of $l_{ij}$ and $u_{ij}$ from line 2(c), we have
    \[
        t_{b-1}^{ij}(\rho_1,\ldots,\rho_n)\in e_{ij}\pm \left(d(d+\sqrt{2})\left[\sum_{m=1}^{b-2}(\sqrt{2}d)^{b-2-m}\epsilon_m\right]+\epsilon_{b-1}\right)n^{b-1}.
    \]
    We conclude that
    \begin{eqnarray}
        t_b(\rho_1,\ldots,\rho_n) &=& \sum_{ij}\trace(\sigma_{j}\rho_{i})t_{b-1}^{ij}(\rho_1,\ldots,\rho_n)\nonumber\\
        &\subseteq&        \left[\sum_{ij}\trace(\sigma_{j}\rho_{i})e_{ij}\right]+ \left(d(d+\sqrt{2})\left[\sum_{m=1}^{b-2}(\sqrt{2}d)^{b-2-m}\epsilon_m\right]+\epsilon_{b-1}\right)\left[\sum_{ij}\trace(\sigma_{j}\rho_{i})\right]n^{b-1}\nonumber\\
        &\subseteq&        \left[\sum_{ij}\trace(\sigma_{j}\rho_{i})e_{ij}\right]+ \sqrt{2}d\left(d(d+\sqrt{2})\left[\sum_{m=1}^{b-2}(\sqrt{2}d)^{b-2-m}\epsilon_m\right]+\epsilon_{b-1}\right)n^{b}\label{eqn:miniclaim2}\\        &\subseteq&        \left[[L,U]\pm \epsilon_{b-1}d^2 n^b\right]+ \sqrt{2}d\left(d(d+\sqrt{2})\left[\sum_{m=1}^{b-2}(\sqrt{2}d)^{b-2-m}\epsilon_m\right]+\epsilon_{b-1}\right)n^{b}\nonumber\\    &\subseteq&        [L,U] \pm d(d+\sqrt{2})\left[\sum_{m=1}^{b-1}(\sqrt{2}d)^{b-1-m}\epsilon_m\right]n^{b}\nonumber,
        \end{eqnarray}
    where the third statement follows from a calculation similar to Eqn.~(\ref{eqn:yetanotherequation}), and the fourth statement from line 3(b) of Alg.~\ref{alg:linearize}. This proves the claim of Eqn.~(\ref{eqn:miniclaim}).

    To complete the proof of Lem.~\ref{l:guarantee}, observe that by Lem.~\ref{l:feasible}, the assignment $\rho^{\operatorname{opt}}$ is feasible for $P_2$ with probability at least $1-d^{2k}n^{k-f}$. Thus, plugging $\rho^{\operatorname{opt}}$ into each of the $d^2n$ linear constraints produced by the recursive calls to LINEARIZE on each $t_{k-1}^{ij}$, we have by Eqns.~(\ref{eqn:miniclaim}) and~(\ref{eqn:miniclaim2}) that (with probability $1-d^{2k}n^{k-f}$)
    \begin{eqnarray*}
        \optprod=t_k(\rho^{\operatorname{opt}})&=&\sum_{ij}\trace\left(\sigma_{j}\rho_{i}^{\operatorname{opt}}\right)t_{k-1}^{ij}(\rho^{\operatorname{opt}})\\
        &\subseteq&\left[\sum_{i,j}\trace(\sigma_{j}\rho_{i}^{\operatorname{opt}})e_{ij}\right]\pm \sqrt{2}d\left(d(d+\sqrt{2})\left[\sum_{m=1}^{k-2}(\sqrt{2}d)^{k-2-m}\epsilon_m\right]+\epsilon_{k-1}\right)n^{k}\\        &\subseteq&\optt\pm d(d+\sqrt{2})\left[\sum_{m=1}^{k-1}(\sqrt{2}d)^{k-1-m}\epsilon_m\right]n^{k},
                \end{eqnarray*}
    where the last statement follows since $\rho^{\operatorname{opt}}$ is not necessarily the optimal solution to $P_2$.
\end{proof}
\end{document}